\documentclass[letterpaper, 10 pt, conference]{ieeeconf}  

\IEEEoverridecommandlockouts                              

 \usepackage{mathrsfs}
 \usepackage{xcolor}
  \usepackage{mathtools}

\usepackage{amsmath,amssymb,amsfonts}
\usepackage{graphicx}
\usepackage{algorithm,algorithmic}

\usepackage{textcomp}

\usepackage{setspace}
\usepackage{lipsum}

\usepackage{mathtools}
\usepackage{enumerate}

\usepackage{caption}
\usepackage{subcaption}
\usepackage{algorithm}
\usepackage{float}

\newtheorem{theorem}{Theorem}

\newtheorem{corollary}{Corollary}

\newtheorem{remark}{Remark}
\newtheorem{proposition}{Proposition}
\usepackage{hyperref}

\makeatletter
\let\NAT@parse\undefined
\makeatother
\usepackage{cite}

\newtheorem{innerexample}{Example}

\newcommand{\abs}[1]{\left\lvert#1\right\rvert}

    {\begin{innerexample}[#1]\pushQED{\qed}}%
    {\popQED\end{innerexample}}
\title{\LARGE \bf
Global Exponential Stabilization of the Kinematic Bicycle Model of a Car in Polar Coordinates
}

\author{Velimir Todorovski$^{1}$, Kwang Hak Kim$^{1}$, Alessandro Astolfi$^{2}$, and Miroslav Krstić$^{1}$ 
\thanks{This work was supported by the Office of Naval Research under Grant No. N00014-23-1-2376 and  N00014-23-1-2831. 
The results and opinions in this paper are solely of the authors and do not reflect the position or the policy of the U.S. Government.}
\thanks{$^{1}$V. Todorovski, K. H. Kim, and M. Krstić are with the Department of Mechanical and Aerospace Engineering, UC San Diego, 9500 Gilman Drive, La Jolla, CA, 92093-0411, {\tt\small \{vtodorovski,kwk001,krstic\}@ucsd.edu}}%
\thanks{$^{2}$A. Astolfi is with the Computer, Electrical and Mathematical Science
and Engineering Division, King Abdullah University of Science and
Technology (KAUST), Thuwal, 23955, Saudi Arabia,
{\tt\small alessandro.astolfi@kaust.edu.sa}}
}

\begin{document}

\maketitle
\thispagestyle{empty}
\pagestyle{empty}

\begin{abstract}
At parking speeds, the kinematic bicycle is the prevailing model for car-like vehicles. Yet, despite its wide use, stabilizing feedback laws for this system are scarce in the literature, and existing designs often do not reproduce realistic parking maneuvers.
This limitation is inherent to the Cartesian coordinates, where Brockett’s condition rules out smooth static feedback stabilization.
We bypass this obstruction by transforming the system into polar coordinates together with additional "range-normalized" coordinates that encode the geometry of human-like parking maneuvers. 
In the transformed coordinates, the dynamics take a strict-feedback form, enabling a nonconventional backstepping design. We exploit the particular structure to develop smooth feedback laws that achieve global exponential stabilization in the transformed coordinates which in turn generates parking trajectories resembling the one performed by human drivers through feedback alone.
\end{abstract}


\section{Introduction}
In this note, we develop globally exponentially \textit{output}-stabilizing control laws for the kinematic bicycle model in polar coordinates, which is widely used to describe the low-speed dynamics of car-like vehicles~\cite{rajamani_vehicle_2012, carvalho2015automated}.
This result is the first of its kind to exploit the geometry of the polar coordinates in order to generate human-like parking trajectories solely through feedback.

In contrast to the common approaches to bicycle control such as model predictive control and related optimization-based methods based on discretized models (see \cite{kong2015kinematic,yu2021model} and the ref. therein), this work is concerned with the continuous time geometric control of the kinematic bicycle.

In our setting, feedback design is constrained by the topological obstructions inherent to nonholonomic systems, which do not arise in discrete models. In particular, although globally controllable, the kinematic bicycle, like other nonholonomic vehicles, cannot be asymptotically stabilized by continuous or discontinuous time-invariant feedback, as shown by Brockett as well as Ryan, Coron, and Rosier \cite{brockett1983asymptotic,ryan1994brockett,coron1994relation}.
%
\begin{figure}[t]
\centering
\includegraphics[width=0.75\linewidth]{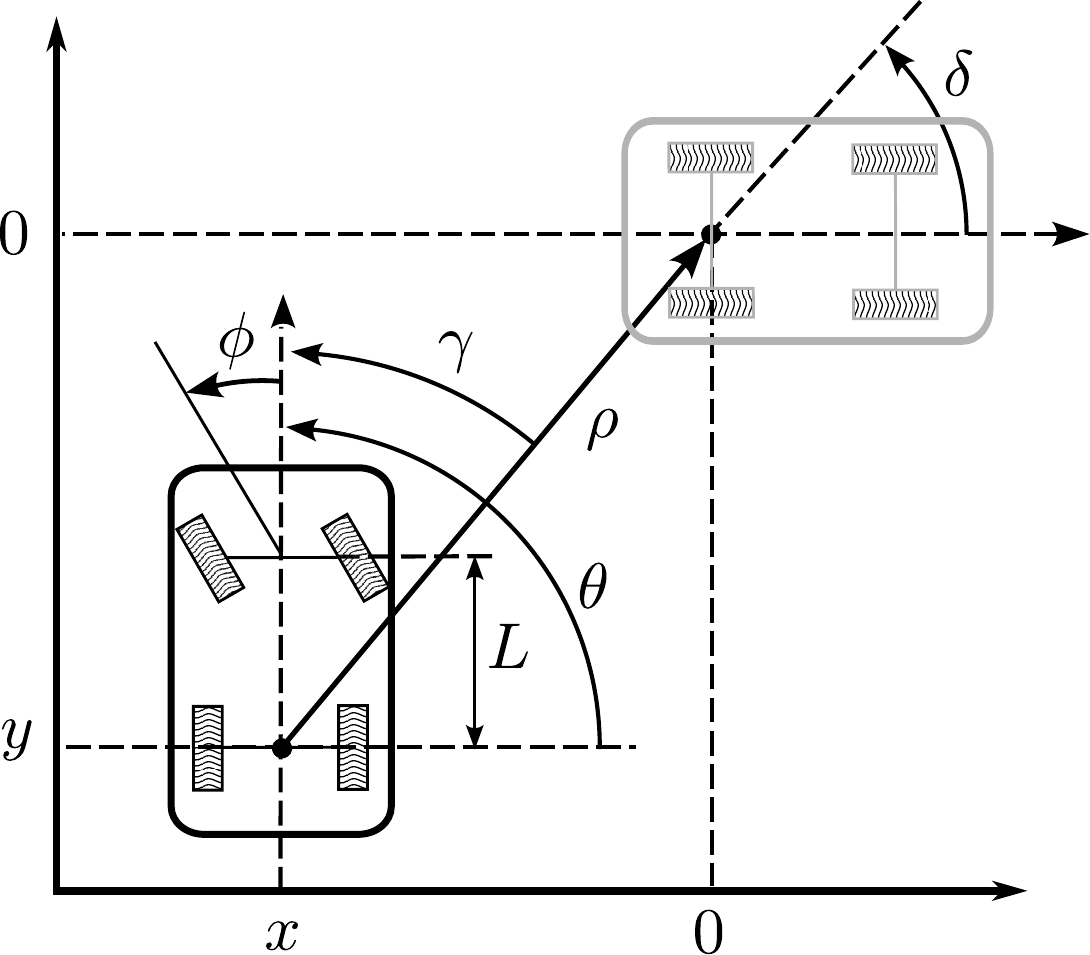}
\caption{Configuration of the  nonholonomic car-like vehicle \((x,y,\theta,\phi)\) with respect to the target \((0,0,0,0)\), and the polar coordinate transformation \((x,y,\theta,\phi)\mapsto (\rho,\delta,\gamma,\varphi)\).}
\label{fig:bicycle_cord}
\end{figure}
A standard way to address this difficulty is to transform the bicycle model from Cartesian coordinates into chained form, as in \cite{MurraySastry1993}, and then design either time-varying or discontinuous feedback laws for the resulting chained system~\cite{astolfi1996discontinuous,teel1995non,de2005feedback,morin2008motion, astolfi1995exponential_car}, thereby circumventing the Brockett--Ryan--Coron--Rosier obstruction. For the bicycle, however, the Cartesian-to-chained transformation is only valid on a restricted domain (see \cite[eqs. (8), (9)]{de2005feedback}), which substantially limits the resulting region of attraction. Moreover, time-varying feedback laws produce oscillatory trajectories and yield slow convergence to the target.
Note that the above approaches inherently rely on the driftless kinematic bicycle model, in which the longitudinal velocity and steering rate are directly actuated. This is not the setting considered in this note,  following \cite{kong2015kinematic,carvalho2015automated,rahman2023driver, black2023future}, we take longitudinal acceleration and steering angle as the inputs, so the bicycle model is not driftless and therefore cannot be converted to an equivalent chained form.

A geometric way to bypass Brockett's obstruction is provided by the use of polar coordinates, which have proved particularly successful in parking problems for the unicycle~\cite{aicardi1995,astolfi1999exponential,restrepo2020leader,wang2026further}. Due to their inherent singularity, they avoid the Brockett obstruction in the Cartesian model, thus enabling the design of smooth feedback laws that are global in the polar coordinates. The advantages of the polar coordinates are further explored in our recent works \cite{todorovski2026nonholonomicrobotparkingfeedback,kim2025nonholonomicrobotparkingfeedback}, where we develop families of control Lyapunov functions (CLFs) for the unicycle that yield inverse-optimal control laws and allow shaping of the cost functions.

However, these advantages do not extend directly to the bicycle model in polar coordinates, since Brockett's condition is only necessary, not sufficient, for the existence of a smooth stabilizer. To address this limitation, we introduce the range-normalized coordinates, which are singular in the spirit of \cite{astolfi1996discontinuous}, encode a human-like parking geometry, and permit the design of smooth control laws in the normalized coordinates through backstepping. Then, viewing the original polar coordinates as outputs of the range-normalized system, we obtain global exponential output stability \cite{karafyllis2024global} in the polar coordinates and almost global exponential attractivity in the Cartesian coordinates. Moreover, the resulting control laws are accompanied by a CLF satisfying the small control property. The design is global except on a measure-zero set corresponding to the vehicle being exactly at the parking position but misaligned with it, a configuration from which, as any driver knows, the vehicle must first move away in order to complete the parking maneuver.


\raggedbottom
\section{Cartesian and Polar Bicycle Models}
\begin{table}[t]

\renewcommand\arraystretch{1.3}
\centering
\begin{tabular}{|l|l|l|}
\hline
\textbf{Polar Coordinates}  & \textbf{Description} \\
\hline
$\rho = \sqrt{x^2 + 
y^2}$  & Distance to origin \\
\hline
$\delta = \text{{\rm atan2}}(y ,x )  + \pi$ & Polar  angle
\\
\hline
$\gamma = \delta - \theta$  & Line-of-sight (LoS) angle   \\
\hline
\end{tabular}
\caption{Polar coordinates and their expressions in terms of Cartesian coordinates. The transformation $(x,y,\theta)\mapsto (\rho,\delta,\gamma)$ is discontinuous on $\{x<0,y=0\}$ and not defined at $x=y=0$. If the target is at $(x^*, y^*, \theta^*)\neq 0$, the transformation generalizes to $\rho=\sqrt{(x-x^*)^2+(y-y^*)^2}, \delta= \text{{\rm atan2}}(y-y^* ,x-x^* ) -\theta^* + \pi, \gamma= \delta-\theta+\theta^*$.}
\label{tab:polar_coordinates}
\end{table}
We consider the following car-like model 
\begin{subequations}
    \label{eq:bicycle_cartesian}
    \begin{align}
        \dot{x} &= v \cos \theta \,, \\
        \dot{y} &= v \sin \theta \,, \\
        \dot{\theta} &= \frac{v}{L} \tan \phi \,, \\
        \dot{v}  &= a \,,
    \end{align}
\end{subequations}
where $(x(t),y(t))$ is the position of the bicycle in Cartesian coordinates, $\theta(t)$ is the heading angle,  $v(t)$ is the forward velocity, $a(t)$ is the longitudinal acceleration input, $\phi(t) \in (-\pi/2,\pi/2)$ is the forward wheel angle which is considered as an input and $L$ is the length of the wheelbase. Unlike the more commonly studied kinematic bicycle, where $v$ is directly actuated and $\phi$ enters through an integrator, \eqref{eq:bicycle_cartesian} is not driftless, which considerably limits the available feedback designs.
The bicycle represented in polar coordinates (cf. Fig.~\ref{fig:bicycle_cord}) is, under the input transformation
   $\phi = \arctan(\varphi)$, 
described by
\begin{subequations}
\label{eq:bicycle_polar}
\begin{align}
    \dot{\rho}&=-v\cos\gamma, \label{eq:dot_rho} \\
    \dot{\delta}&=\frac{v}{\rho}\sin\gamma, \label{eq:dot_delta} \\
    \dot{\gamma}&=\frac{v}{\rho}\sin\gamma-\frac{v}{L}\varphi\,,
    \label{eq:dot_gamma} \\
    \dot{v}&=a.
\end{align}
\end{subequations}
The input transformation maps the unconstrained auxiliary input $\varphi \in \mathbb{R}$ into the physically admissible steering range $\phi \in (-\pi/2,\pi/2)$. The corresponding definitions of the polar coordinates are provided in Table~\ref{tab:polar_coordinates}. Since the polar coordinates are only valid on the set $\{\rho > 0\}$, we introduce the state space
\begin{align}
\mathcal{S} := \left\{ \rho > 0 \right\} \times \mathbb{R}^3,
\label{eq:ss_S}
\end{align}
on which the states $(\rho,\delta,\gamma,v) \in \mathcal{S}$ evolve 
and for each $(\rho,\xi) \in \mathcal{S}$ with $\xi = (\xi_1,\xi_2,\xi_3)$, we define the norm
\begin{align}
|(\rho,\xi)|_{\mathcal{S}}
:= \rho + |\xi_1| + |\xi_2| + |\xi_3|.
\label{eq:metric_S}
\end{align}
In Cartesian coordinates, the Brockett--Ryan--Coron--Rosier necessary conditions rule out continuous time-invariant stabilizers for the bicycle model \eqref{eq:bicycle_cartesian}. In the polar coordinates (cf. Tab.~\ref{tab:polar_coordinates}), these conditions do not apply, since these coordinates are only defined on $\{\rho>0\}$ and \eqref{eq:bicycle_polar} has a singularity at $\rho=0$. Thus, the polar formulation removes this particular obstruction, but by itself does not imply the existence of a smooth static stabilizer. In what follows, we address this question.
\section{Backstepping Through Polar and Range-Normalized Coordinates}
The controller construction proceeds in two steps. First, we reformulate \eqref{eq:bicycle_polar} so that the strict-feedback structure in the polar coordinates becomes explicit. In particular, since the LoS angle $\gamma$ enters both $\dot{\rho}$ and $\dot{\delta}$, the system can be viewed as a triple-integrator-like chain modulated by the trigonometric terms $\sin\gamma$ and $\cos\gamma$, which limit the magnitude of the inputs that acts on $\dot{\rho}$ and $\dot{\delta}$. We then introduce range-normalized coordinates to encode how the angular variables and the velocity must scale with the remaining distance in a human-like parking maneuver, and use this structure to construct the control laws via backstepping.

\subsection{Strict-feedback form  of the polar bicycle}
Central to our backstepping approach are the bounded functions
\begin{align}
\psi_1(r,s) &= \frac{\cos(r-s)-\cos(s)}{r}, \label{eq:psi_1_definition}\\
\psi_2(r,s) &= \frac{\partial \psi_1(r,s)}{\partial s}
= \frac{\sin(r-s)+\sin(s)}{r},
\label{eq:psi_2_definition}
\end{align}
which are smooth, i.e., $C^\infty$ in $(r,s)$ when interpreted via their smooth extensions at $r=0$. In particular, these extensions satisfy
$\psi_1(0,s)=\sin s$
and $\psi_2(0,s)=\cos s$, obtained by taking the limits as $r\to 0$.
Moreover, $|\psi_1(r,s)|\le 1$ and $|\psi_2(r,s)|\le 1$ for all $(r,s)$, and
$\psi_1(0,0)=0$, $\psi_2(0,0)=1$ and
$\psi_1(0,2n\pi) = 0$, $\psi_2(0,2n\pi)=1$ for all $n\in\mathbb{Z}$. With \eqref{eq:psi_1_definition} and \eqref{eq:psi_2_definition}, the polar bicycle \eqref{eq:bicycle_polar} can be rewritten as
\begin{subequations}
\label{eq:system_in_backstep1}
\begin{align}
    \dot{\rho} &= -v \cos(z-\gamma) + v \psi_1(z,\gamma) z, \\
    \dot{\delta} &= -\frac{v}{\rho} \sin(z-\gamma) + \frac{v}{\rho} \psi_2(z,\gamma) z, \\
    \dot{\gamma} &= -\frac{v}{\rho} \sin(z-\gamma) + \frac{v}{\rho} \psi_2(z,\gamma) z - \frac{v}{L}\varphi,  \\
    \dot{v} &= a\,,
\end{align}
\end{subequations}
where $z$ is an auxiliary variable introduced for the backstepping design. To exploit the structure of \eqref{eq:system_in_backstep1}, in addition to \eqref{eq:psi_1_definition} and \eqref{eq:psi_2_definition}, the following trigonometric identities play an equally important role
\begin{equation}
\cos(\arctan(s)) = \frac{1}{N(s)}, 
\quad
\sin(\arctan(s)) = \frac{s}{N(s)},
\end{equation}
where
\begin{equation}
N(s) = \sqrt{1+s^2}\,, \label{eq:N(s)}
\end{equation}
and hence \(N(s) \ge 1\) for all \(s \in \mathbb{R}\).
Then, analogous to the exponentially stabilizing designs for the unicycle   \cite{kim2025nonholonomicrobotparkingfeedback}, 
we introduce the backstepping transformation 
\begin{equation}
 z = \gamma + \arctan(k_2\delta), \quad k_2 > 0\,. \label{eq:backstepping_z}
\end{equation}
Substituting \eqref{eq:backstepping_z} in \eqref{eq:system_in_backstep1}, we obtain
\begin{subequations}
\label{eq:system_in_backstep2}
\begin{flalign}
    \dot{\rho} &=
    -\frac{v}{N(k_2\delta)} + v \psi_1(z,\gamma) z, &\\
    \dot{\delta} &=
    -\frac{v}{\rho} \frac{k_2 \delta}{N(k_2\delta)}
    + \frac{v}{\rho} \psi_2(z,\gamma) z, & \\
    \dot{z} &= - \frac{v}{L} \varphi
    - \frac{v}{\rho}
    \left[1 + \frac{k_2}{N^2(k_2\delta)}\right]   \left[
        \frac{k_2 \delta}{N(k_2\delta)} - \psi_2(z,\gamma) z 
    \right],  &\\
    \dot{v} &= a.
\end{flalign}
\end{subequations}
Observe that the dynamics of \eqref{eq:system_in_backstep2} are fully determined by the variables $(\rho,\delta,z,v)$, since the LoS angle $\gamma$ is recovered from
$\gamma = z - \arctan(k_2\delta),$
which follows from \eqref{eq:backstepping_z}.
If the velocity $v$ was directly actuated, choosing $v= N(\delta) \rho $ would result in locally exponentially stable $(\rho,\delta)$-dynamics. Even though this is not possible since the velocity $v$ is now a state, this provides an insight on how to choose the backstepping transformation for it.
To this end, we introduce the second backstepping transformation as
\begin{equation}
 \hat{v} = v - k_1 N(k_2 \delta) \rho , \quad k_1 > 0. \label{eq:backstepping_v_hat}
\end{equation} 
Now, choosing the longitudinal acceleration as
\begin{equation}
\hspace*{-0.2cm}
\begin{aligned}[b]
a &= \tilde{a} - k_1\bigl(\hat v+k_1N(k_2\delta)\rho\bigr)\Bigg[
1+\frac{k_2^3\delta^2}{N(k_2\delta)^2} \\
&\qquad -N(k_2\delta)\left(
\psi_1(z,\gamma)
+\frac{k_2^2\delta}{N^2(k_2\delta)}\psi_2(z,\gamma)
\right)z
\Bigg] \label{eq:acceleration_first}
\end{aligned}\,,
\end{equation}
and substituting \eqref{eq:backstepping_v_hat}  in \eqref{eq:system_in_backstep2}, we obtain
{%
\allowdisplaybreaks[3]
\begin{subequations}
\label{eq:polar_bicycle_nonsingular}
\begin{align}
\dot{\rho}
&= -k_1 \rho + k_1 N(k_2 \delta)\rho \psi_1(z,\gamma)z \nonumber\\
&\hspace*{1.26cm} + \left[\psi_1(z,\gamma)z - N(k_2 \delta)^{-1}\right]\hat{v},
\displaybreak[3]\\
\dot{\delta}
&= -k_1 k_2 \delta + k_1 N(k_2 \delta)\psi_2(z,\gamma)z \nonumber\\
&\qquad \qquad \hspace*{0.19cm}+ \frac{\hat{v}}{\rho}\left[\psi_2(z,\gamma)z - k_2 N(k_2 \delta)^{-1}\delta\right],
\displaybreak[3] \label{eq:delta_dot_nonsingular}\\
\dot{z}
&= \frac{\hat{v} + k_1 N(k_2 \delta)\rho}{\rho}
\left[1 + \frac{k_2 }{N^2(k_2\delta)}\right] \times \nonumber\\
& \quad \left[\psi_2(z,\gamma)z -  \frac{k_2\delta}{N(k_2 \delta)}\right] - \frac{\hat{v} + k_1 N(k_2 \delta)\rho}{L}\varphi, \label{eq:z_dot}
\displaybreak[3]\\
\dot{\hat v}
&= \tilde{a}.
\end{align}
\end{subequations}
}%
Note that if the steering input $\varphi$ were decoupled from the velocity $\hat{v}$ in \eqref{eq:z_dot}, that is, if the vehicle were allowed to turn in place, then one could choose $\varphi$ and $\tilde a$ so that the destabilizing residual terms in $\dot{\rho}$ and $\dot{\delta}$ be cancelled and thereby obtaining exponentially stabilizing dynamics. For the kinematic bicycle, however, such a decoupling is not possible. We address this issue next.

\subsection{Human parking geometry in normalized coordinates}
Since the car-like bicycle can change its orientation only while moving, a human parking motion must simultaneously align the vehicles heading as it decreases the distance to the target, not after the target is reached. Accordingly, the angular errors $\gamma$ and $\delta$ must decay to zero no slower than the remaining distance $\rho$. 
This motivates the introduction of the range-normalized angles
\begin{align}
 \bar{\delta} &= \frac{\delta}{\rho} \,, \label{eq:bar_delta} \\
    \bar{z} &= \frac{z}{\rho} = \frac{\gamma}{\rho} + \frac{\arctan(k_2 \rho\bar{\delta})}{\rho} \,,
    \label{eq:bar_z}
\end{align}
which encode the geometry of a human-like parking maneuver. Indeed, as the distance-to-go $\rho$ becomes small, the polar angle $\delta$ and the LoS angle $\gamma$ must vanish not slower than $\rho$. In physical terms, this means that near the target the car is already nearly aligned before the last remaining distance is covered.  Note that, on the state-space \eqref{eq:ss_S}, it holds that $\rho>0$ and hence this normalization is well defined and does not impose any additional restriction on the region of attraction $\mathcal{S}$. Moreover, the backstepping stabilizing term $\arctan(k_2 \rho\bar{\delta})/\rho$ in \eqref{eq:bar_z} tends to zero as $\rho$ goes to zero.
Next, we introduce the range-normalized velocity
\begin{equation}
    \bar{v}=\frac{\hat{v}}{\rho}
    =\frac{v}{\rho}-k_1N\bigl(k_2\rho\bar{\delta}\bigr) \label{eq:bar_v}
\end{equation}
to measure the speed relative to the remaining  to the target.
This normalization is needed because the angular dynamics \eqref{eq:delta_dot_nonsingular} and \eqref{eq:z_dot} depend on the ratio $\hat{v}/\rho$, so near the target the relevant quantity is not the speed $v$ itself, but how fast it decays relative to $\rho$. 
It also fits our physical intuition, reflecting the fact that in a human parking motion the car slows down as it approaches the target, so that the remaining steering alignment can be completed before the last bit of distance is covered. Accordingly, near the parking spot, the speed must already be small relative to the remaining distance, rather than remaining excessively large as $\rho \to 0$.
Then, let the acceleration and steering inputs be 
\begin{equation}
\begin{aligned}[b]
    \tilde{a} &= \rho \Biggl[ \hat{a} -  k_1\bar{v}
+k_1N(k_2\rho\bar\delta)\psi_1(\rho\bar z,\gamma)\,\rho\bar z \bar{v}\\
&\qquad \quad \qquad \hspace*{0.23cm}\left.+\left[\psi_1(\rho\bar z,\gamma)\,\rho\bar z-\frac{1}{N(k_2\rho\bar\delta)}\right]\bar v^2
\right]\,, 
    \end{aligned}
\label{eq:acceleration_second}
\end{equation}
and
{
\allowdisplaybreaks
\begin{align}
&\varphi =
\frac{L}{k_1 N(k_2 \rho \bar{\delta})}
\Biggl\{
-\tilde{\varphi}
+k_1\bar z
-k_1 \rho N(k_2\rho\bar\delta)\psi_1(\rho\bar z,\gamma)\bar z^2
\nonumber\\
&\hspace*{-0.2cm}+\left(k_1+\frac{k_1 k_2}{N^2(k_2\rho\bar\delta)}\right)
\left(
N(k_2\rho\bar\delta)\psi_2(\rho\bar z,\gamma)\bar z
-k_2\bar\delta
\right)
\nonumber\\
&\hspace*{-0.2cm}+ \Biggl[
\left(
N(k_2\rho\bar\delta)^{-1}
-\rho\bar z\,\psi_1(\rho\bar z,\gamma)
\right)\bar z
\nonumber\\
&\hspace*{-0.2cm}+ \left(1+\frac{k_2}{N^2(k_2\rho\bar\delta)}\right)
\left(
\psi_2(\rho\bar z,\gamma)\bar z
-\frac{k_2\bar\delta}{N(k_2\rho\bar\delta)}
\right)
\Biggr]\bar v
\Biggr\},
\label{eq:steering_1}
\end{align}
}
respectively. Substituting \eqref{eq:bar_delta}, \eqref{eq:bar_z}, and \eqref{eq:bar_v} into \eqref{eq:polar_bicycle_nonsingular} yields 
\begin{subequations}
\label{eq:polar_bicycle_singular}
\begin{align}
\dot{\rho}
&=
-k_1\rho
+k_1N(k_2\rho\bar\delta)\psi_1(\rho\bar z,\gamma)\,\rho^2\bar z \nonumber\\
&\qquad \qquad \qquad 
+\left[\psi_1(\rho\bar z,\gamma)\,\rho^2\bar z-\frac{\rho}{N(k_2\rho\bar\delta)}\right] \bar v,
\\[2mm]
\hspace*{-0.2cm}\dot{\bar\delta}
&=
-k_1(k_2-1)\bar\delta  \\ &\quad +k_1N(k_2\rho\bar\delta)\Bigl[\psi_2(\rho\bar z,\gamma)-\rho\bar\delta\,\psi_1(\rho\bar z,\gamma)\Bigr]\bar z \nonumber \\
&
\hspace*{0.2cm}+\Biggl[
\bigl(\psi_2(\rho\bar z,\gamma)-\rho\bar\delta\,\psi_1(\rho\bar z,\gamma)\bigr)\bar z  +\frac{(1-k_2) \bar\delta}{N(k_2\rho\bar\delta)}
\Biggr] \bar{v},
\\[1mm]
\dot{\bar z}
&= \tilde{\varphi} 
 - \frac{\varphi}{L}\bar{v} ,
\\
\dot{\bar v}
&= \hat{a}.
\end{align}
\end{subequations}
whose range-normalized states $(\rho,\bar{\delta},\bar{z},\bar{v})$ evolve on $\mathcal{S}$.

\section{Exponential Behavior in Range-Normalized, Polar, and Cartesian Coordinates}
\subsection{Global exponential stability in normalized coordinates }
\begin{theorem}
\label{thm:normalized_coordinates_exp_stability}
    Consider the system \eqref{eq:polar_bicycle_singular} in closed-loop with 
    \begin{align}
        \tilde{\varphi} =& -k_3 \bar{z} - k_1N(k_2\rho\bar\delta)\psi_1(\rho\bar z,\gamma)\,\rho^3 \nonumber \\
        &\qquad  -k_1N(k_2\rho\bar\delta)\Bigl[\psi_2(\rho\bar z,\gamma)-\rho\bar\delta\,\psi_1(\rho\bar z,\gamma)\Bigr] \bar{\delta} \,,\label{eq:varphi_thm}
    \end{align}
    and
    \begin{equation}
    \hspace*{-0.3cm}
    \begin{aligned}[b]
        &\hat{a} = -k_4 \bar{v} - \left[\psi_1(\rho\bar z,\gamma)\,\rho^2\bar z-\frac{\rho}{N(k_2\rho\bar\delta)}\right] \rho  \\
        &-\Biggl[
\bigl(\psi_2(\rho\bar z,\gamma)-\rho\bar\delta\,\psi_1(\rho\bar z,\gamma)\bigr)\bar z  +\frac{(1-k_2) \bar\delta}{N(k_2\rho\bar\delta)}
\Biggr] \bar{\delta}  +\frac{\varphi}{L}\bar{z}\,, \label{eq:hat_a_thm}
    \end{aligned}
\end{equation}
    with $k_1,k_3,k_4 > 0$ and $k_2 > 1$, where $\varphi(t)$ is defined in \eqref{eq:steering_1} and the functions $\psi_1(r,s)$, $\psi_2(r,s)$ and $N(s)$ are defined in \eqref{eq:psi_1_definition}, \eqref{eq:psi_2_definition} and \eqref{eq:N(s)}, respectively. Then, there exists $K, \lambda > 0$ such that
    \begin{equation}
        \hspace*{-0.2cm} \abs{(\rho(t),\bar{\delta}(t), \bar{z}(t), \bar{v}(t))}_{\mathcal{S}} \le K \abs{(\rho_0,\bar{\delta}_0, \bar{z}_0, \bar{v}_0)}_{\mathcal{S}} {\rm e}^{-\lambda (t-t_0)}, \label{eq:exp_bound_singular_coordinates}
    \end{equation}
    for all $t\ge t_0$. 
\end{theorem}
\begin{proof}
    Consider the Lyapunov function 
    \begin{equation}
        V = \frac{1}{2}\left(\rho^2 + \bar{\delta}^2 + \bar{z}^2 + \bar{v}^2\right). \label{eq:CLF_singular}
    \end{equation}
    Its time-derivative of \eqref{eq:CLF_singular} along the solutions of \eqref{eq:polar_bicycle_singular} is 

{\allowdisplaybreaks
\begin{align}
\dot V
&= -k_1 \rho^2 -k_1(k_2-1)\bar{\delta}^2
+ \bar{z} \Biggl\{ \tilde{\varphi}
+ k_1N(k_2\rho\bar\delta)\psi_1(\rho\bar z,\gamma)\,\rho^3
\nonumber\\
&\hspace*{-0.3cm}+ k_1N(k_2\rho\bar\delta)\Bigl[\psi_2(\rho\bar z,\gamma)-\rho\bar\delta\,\psi_1(\rho\bar z,\gamma)\Bigr]\bar{\delta}
\Biggr\}
\nonumber\\
&\hspace*{-0.3cm}+ \bar{v}\Biggl\{ \hat{a}
+ \left[\psi_1(\rho\bar z,\gamma)\,\rho^2\bar z-\frac{\rho}{N(k_2\rho\bar\delta)}\right]\rho
+\Biggl[ \frac{(1-k_2)\bar\delta}{N(k_2\rho\bar\delta)}
\nonumber\\ 
&\qquad  +
\bigl(\psi_2(\rho\bar z,\gamma)-\rho\bar\delta\,\psi_1(\rho\bar z,\gamma)\bigr)\bar z
\Biggr]\bar{\delta}
-\frac{\varphi}{L}\bar z
\Biggr\}.
\label{eq:V_dot_gen}
\end{align}
}
Then, substituting \eqref{eq:varphi_thm} and \eqref{eq:hat_a_thm} in \eqref{eq:V_dot_gen}, yields
\begin{align}
    \dot V = -k_1 \rho^2 - k_1(k_2-1) \bar{\delta}^2 - k_3 \bar{z}^2 - k_4 \bar{v}^2 \,, \label{eq:dot_V_exp_sing}
\end{align}
which is negative when $k_2 > 1$ for all $(\rho, \bar{\delta}, \bar{z}, \bar{v}) \ne (0,0,0,0)$. Furthermore, 
considering \eqref{eq:CLF_singular} from \eqref{eq:dot_V_exp_sing}, we write
\begin{equation}
    \dot V \le - \underline{c} V 
\end{equation}
where $\underline{c} =2\min(k_1,k_1(k_2-1), k_3, k_4)$ from which \eqref{eq:exp_bound_singular_coordinates} follows, where $K = 2$ and $\lambda = \underline{c}/2$.
\end{proof}
\begin{remark}
The estimate \eqref{eq:exp_bound_singular_coordinates} in Thm.~\ref{thm:normalized_coordinates_exp_stability} implies forward completeness of the solution $(\rho,\bar{\delta},\bar{z}, \bar{v})$, as well as global exponential stability (in the $\mathcal{KL}$ sense) of the origin
on $\mathcal{S}$. The normalized variables introduced in \eqref{eq:bar_delta}--\eqref{eq:bar_v} remain well defined along every trajectory starting from $\mathcal S$. 
In particular, although the change of variables is singular at $\rho=0$, the estimate \eqref{eq:exp_bound_singular_coordinates} implies that $|\bar{\delta}(t)|$, $|\bar{z}(t)|$, and $|\bar{v}(t)|$ are each bounded from above by the norm $\abs{(\rho(t),\bar{\delta}(t),\bar{z}(t),\bar{v}(t))}_{\mathcal S}$, and therefore decay exponentially for all $t\ge t_0$. Since the trajectories evolve on $\mathcal S=\{\rho>0\}\times\mathbb R^3$, one has $\rho(t)>0$ for every finite $t\ge t_0$, so the singular boundary $\rho=0$ is never reached in finite time and the normalized states $(\rho(t), \bar{\delta}(t), \bar{z}(t), \bar{v}(t))$ remain bounded for any finite $t\ge t_0$.
Moreover, the quadratic function \eqref{eq:CLF_singular} is a global strict CLF for  \eqref{eq:polar_bicycle_singular} on $\mathcal S$. 
\end{remark}

\subsection{Global exponential output stability in polar coordinates}
\begin{proposition}
\label{prop:output_stab}
    Consider \eqref{eq:bicycle_polar} in closed-loop with \eqref{eq:acceleration_first}, \eqref{eq:acceleration_second}, \eqref{eq:steering_1}, \eqref{eq:varphi_thm} and \eqref{eq:hat_a_thm} and let $\Theta := (\rho,\bar{\delta},\bar{z},\bar{v})^\top$. Then, there exists $C, C_{\varphi}, C_{a}, \lambda > 0$ such that the estimates
    \begin{equation}
\abs{(\rho(t),\delta(t),\gamma(t),v(t))}_{\mathcal S} \le C  \Bigl( 1 +  \abs{\Theta_0}^2_{\mathcal S} \Bigr)
   |\Theta_0|_{\mathcal S} e^{-\lambda (t-t_0)}\,, \label{eq:polar_bound}
    \end{equation}
    and 
    \begin{align}
        \abs{\varphi(t)} &\le C_{\varphi} \left(1 + \abs{\Theta_0}_{\mathcal{S}}^4\right) \abs{\Theta_0}_{\mathcal{S}} {\rm e}^{-\lambda (t-t_0)} \,, \label{eq:varphi_bound} \\
        \abs{a(t)} &\le C_{a} \left(1 + \abs{\Theta_0}_{\mathcal{S}}^8\right) \abs{\Theta_0}_{\mathcal{S}} {\rm e}^{-\lambda (t-t_0)} \,, \label{eq:a_bound}
    \end{align}
    hold on $\mathcal{S}$.
\end{proposition}
\begin{proof}
Let us introduce the shorthand notation for the states $Y := (\rho,\delta,\gamma,v)^\top$, then 
the corresponding norms of $\Theta$ and $Y$ on $\mathcal{S}$ with \eqref{eq:metric_S} are 
\begin{equation}
|\Theta|_{\mathcal S} := \rho + |\bar{\delta}| + |\bar{z}| + |\bar{v}|\,, \quad |Y|_{\mathcal S} := \rho + |\delta| + |\gamma| + |v|. 
\end{equation}
From \eqref{eq:bar_z} and \eqref{eq:bar_v}, 
it follows that
\begin{equation}
    \gamma = \rho \bar{z} - \arctan(k_2 \rho \bar{\delta}), \label{eq:gamma_from_X} 
\end{equation}
\begin{equation}
    v = \rho \bar{v} + k_1 \rho N(k_2 \rho \bar{\delta}). \label{eq:v_from_X}
\end{equation}
Since $\rho \le |\Theta|_{\mathcal S}$, $|\bar{\delta}| \le |\Theta|_{\mathcal S}$,
$|\bar{z}| \le |\Theta|_{\mathcal S}$, $|\bar{v}| \le |\Theta|_{\mathcal S}$,
and $|\arctan(s)| \le |s|$, as well as
$N(s) = \sqrt{1+s^2} \le 1 + |s|$,
it follows from \eqref{eq:bar_delta},  that
\begin{align}
|\delta|
= \rho |\bar{\delta}|
\le |\Theta|_{\mathcal S}^2. \label{eq:delta_bound}
\end{align}
Also, by \eqref{eq:gamma_from_X},
\begin{equation}
|\gamma|
\le \rho |\bar{z}| + |\arctan(k_2 \rho \bar{\delta})| 
\le (1+k_2)|\Theta|_{\mathcal S}^2. \label{eq:gamma_bound}
\end{equation}
Finally, by \eqref{eq:v_from_X},
\begin{align}
|v|
&\le \rho |\bar{v}| + k_1 \rho N(k_2 \rho \bar{\delta}) \le \rho |\bar{v}| + k_1 \rho \bigl(1 + k_2 \rho |\bar{\delta}|\bigr) \nonumber \\
&\le |\Theta|_{\mathcal S}^2 + k_1 |\Theta|_{\mathcal S} + k_1 k_2 |\Theta|_{\mathcal S}^3. \label{eq:v_bound}
\end{align}
Therefore, combining \eqref{eq:delta_bound}, \eqref{eq:gamma_bound} and \eqref{eq:v_bound}, we obtain
\begin{equation}
|Y|_{\mathcal S}
\le (1+k_1)|\Theta|_{\mathcal S} + (3+k_2)|\Theta|_{\mathcal S}^2 + k_1 k_2 |\Theta|_{\mathcal S}^3.
\end{equation}
Equivalently,
\begin{equation}
|Y|_{\mathcal S}
\le \underline{C} \bigl(1 + |\Theta|_{\mathcal S}^2\bigr)|\Theta|_{\mathcal S}, \label{eq:temp}
\end{equation}
with $\underline{C} =  2\max\{1+k_1,\; 3+k_2,\; k_1 k_2\}$. Note that \eqref{eq:exp_bound_singular_coordinates} can be expressed as $|\Theta(t)|_{\mathcal S}
\le K |\Theta_0|_{\mathcal S} e^{-\lambda (t-t_0)}$. 
Then, from \eqref{eq:temp} and since $K = 2$ we obtain
\begin{equation}
|Y(t)|_{\mathcal S} \le 8\underline{C}  \Bigl( 1 +  |\Theta_0|^2_{\mathcal S} \Bigr)
   |\Theta_0)|_{\mathcal S} e^{-\lambda (t-t_0)}\,,
\end{equation}
from which \eqref{eq:polar_bound} follows.
 Furthermore, considering the fact that  $N(s)^{-1} \le 1$, $\abs{\psi_1(r,s)} \le 1$ and $\abs{\psi_2(r,s)} \le 1$  from \eqref{eq:steering_1} and \eqref{eq:varphi_thm}, after some lengthy calculations, we obtain
\begin{equation}
|\varphi|
\le \tilde{C_\varphi} \bigl(1+|\Theta|_{\mathcal S}^4\bigr)|\Theta|_{\mathcal S},
\end{equation}
where $\tilde C_{\phi}:=
L\left(
\frac{k_3}{k_1}
+6+5k_2+k_2^2
+\frac{3+2k_2+k_2^2}{k_1}
\right)$.
Likewise, using \eqref{eq:acceleration_first}, \eqref{eq:acceleration_second}, and \eqref{eq:hat_a_thm}, one obtains
\begin{equation}
|a|
\le \tilde{C}_\alpha \bigl(1+|\Theta|_{\mathcal S}^8\bigr)|\Theta|_{\mathcal S},
\end{equation}
where 
$\tilde C_{\alpha}:=
2k_4+2k_2+8+2k_1+k_1k_2
+4\left(
k_3/k_1
+6+5k_2+k_2^2
+(3+2k_2+k_2^2)/k_1
\right)
+2k_1(1+k_1+k_1k_2)(2+k_2+k_2^2+3k_2^3)$.
Then, it follows from \eqref{eq:exp_bound_singular_coordinates} that
\begin{equation}
|\varphi(t)|
\le
32 
\tilde{C}_\varphi
\bigl(1+|\Theta_0|_{\mathcal S}^4\bigr)
|\Theta_0|_{\mathcal S}
e^{-\lambda (t-t_0)},
\end{equation}
and
\begin{equation}
|a(t)|
\le
512 \tilde{C}_\alpha
\bigl(1+|\Theta_0|_{\mathcal S}^8\bigr)
|\Theta_0|_{\mathcal S}
e^{-\lambda (t-t_0)}.
\end{equation}
from which \eqref{eq:varphi_bound} and \eqref{eq:a_bound} follow.
\end{proof}
\begin{remark}
   If $Y(t)=(\rho(t),\delta(t),\gamma(t),v(t))^\top$ is viewed as the output of the range-normalized system \eqref{eq:polar_bicycle_singular} with the feedback laws \eqref{eq:varphi_thm} and \eqref{eq:hat_a_thm}, whose state is $\Theta(t)=(\rho(t),\bar{\delta}(t),\bar{z}(t),\bar{v}(t))^\top$, then the estimate \eqref{eq:polar_bound} in Proposition~\ref{prop:output_stab} corresponds to the property of \textit{uniform global exponential output stability} on $\mathcal S$ (cf. the def. of UGAOS in \cite{karafyllis2024global}).
   This is the best possible outcome since $\Theta(t)$ cannot be globally bounded in terms of $Y(t)$, as $\Theta(t)$ contains the range-normalized states $\delta(t)/\rho(t)$, $\gamma(t)/\rho(t)$, and $v(t)/\rho(t)$, which may blow up as $\rho(t)\to0^+$ even when $Y(t)$ remains bounded (e.g., $\rho_0=\varepsilon$, $\delta_0=1$, $\gamma_0=v_0=0$ gives $|Y_0|_{\mathcal S}=\varepsilon+1$ while $|\Theta_0|_{\mathcal S}\ge 1/\varepsilon\to\infty$).
   Nevertheless, the bounds \eqref{eq:varphi_bound} and \eqref{eq:a_bound} show that the inputs $a(t)$ and $\varphi(t)$ satisfy the small control property in the sense that they both remain bounded and converge to zero as $\Theta(t)\to 0$. This does not exclude large transient control values for certain initial conditions. In particular, if $\rho_0$ is very small while $v_0$ is large, then the resulting braking and steering demands may be correspondingly large. This reflects the extreme case of approaching the parking position at high speed, where large actuation is unavoidable to decelerate and complete the maneuver.
\end{remark}

\begin{figure}
    \centering
    \includegraphics[width=0.9\linewidth]{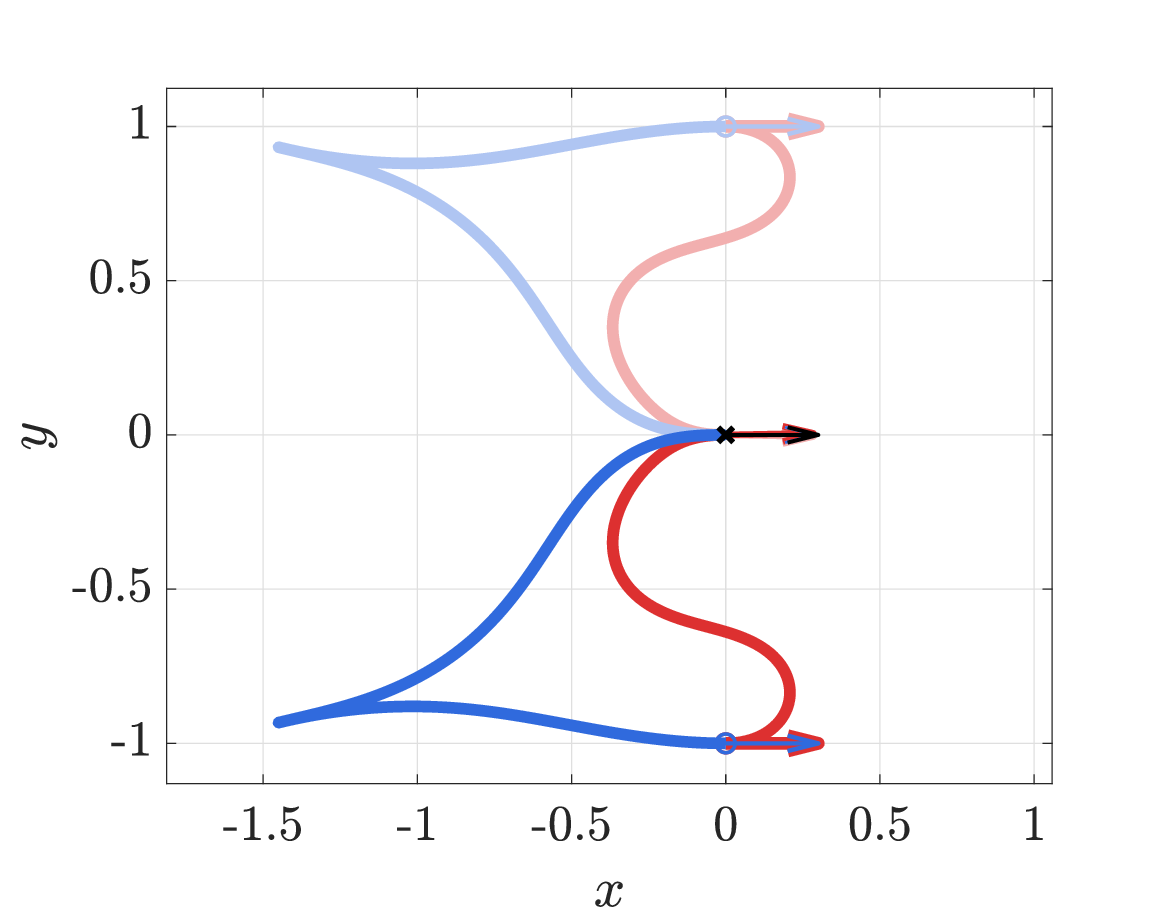}
    \caption{Parallel parking trajectories of \eqref{eq:bicycle_cartesian} with wheelbase 
$L = 4.5$, under the acceleration control law \eqref{eq:acceleration_first}, \eqref{eq:acceleration_second}, \eqref{eq:hat_a_thm} and the steering control law \eqref{eq:steering_1}, \eqref{eq:varphi_thm}. }
    \label{fig:cartesian_pp_trajectories}
\end{figure}
\subsection{Attractivity sans stability in Cartesian coordinates}
\begin{corollary} \label{eq:corollary}
Consider \eqref{eq:bicycle_cartesian} in closed-loop with
\eqref{eq:acceleration_first}, \eqref{eq:acceleration_second},
\eqref{eq:steering_1}, \eqref{eq:varphi_thm}, and \eqref{eq:hat_a_thm}.
For every initial condition $(x_0,y_0,\theta_0,v_0)\in\mathbb{R}^4$
such that $x_0^2+y_0^2>0$, there exist constants $C_0(k_1,k_2)>0$ and $\lambda>0$ such that
\begin{equation}
|x(t)|+|y(t)|+|\theta(t)|+|v(t)|
\le
C_0(1+\abs{X_0}^2)\abs{X_0}e^{-\lambda(t-t_0)}\,, \label{eq:cartesian_estimate}
\end{equation}
for all $ t\ge t_0$, where 
\begin{equation}
X_0:=
\frac{x_0^2+y_0^2+|\theta_0|+|v_0|+1}{\sqrt{x_0^2+y_0^2}}.
\end{equation}
\end{corollary}
\vspace{0.3\baselineskip}

\begin{remark}
The corollary follows from Prop.~\ref{prop:output_stab}, together with $|x|+|y|+|\theta|+|v|
\le
\sqrt{2}\,|(\rho,\delta,\gamma,v)|_{\mathcal S}$, and the bounds $|\bar{\delta}_0|\le c_1 X_0$,
$|\bar{z}_0|\le c_2 X_0$,
$|\bar{v}_0|\le c_3 X_0$ which yield $|(\rho_0,\bar\delta_0,\bar z_0,\bar v_0)|_{\mathcal S}
\le c_4 X_0$
for some  $c_1, c_2, c_3, c_4>0$.
The estimate \eqref{eq:cartesian_estimate} is not a stability estimate in Cartesian coordinates. Rather, it characterizes the Cartesian equilibrium by global Lagrange stability, global uniform exponential attractivity, and Lyapunov instability. This combination of properties is referred to as global uniform Lagrange asymptotic stability (\emph{GULAS}). Here, Lyapunov instability is a structural property of nonholonomic systems subject to the Brockett--Ryan--Coron--Rosier conditions. In particular, trajectories are globally bounded and converge exponentially to the equilibrium, but the equilibrium is not Lyapunov stable in the Cartesian variables. This is not in conflict with the implication from GULAS to GAS for Lipschitz systems~\cite{karafyllis2011stability}, since under the present feedback laws the Cartesian closed-loop system is not Lipschitz.

\end{remark}
\begin{figure}[t]
\centering
    \hspace*{-0.5cm}\includegraphics[scale=0.27]{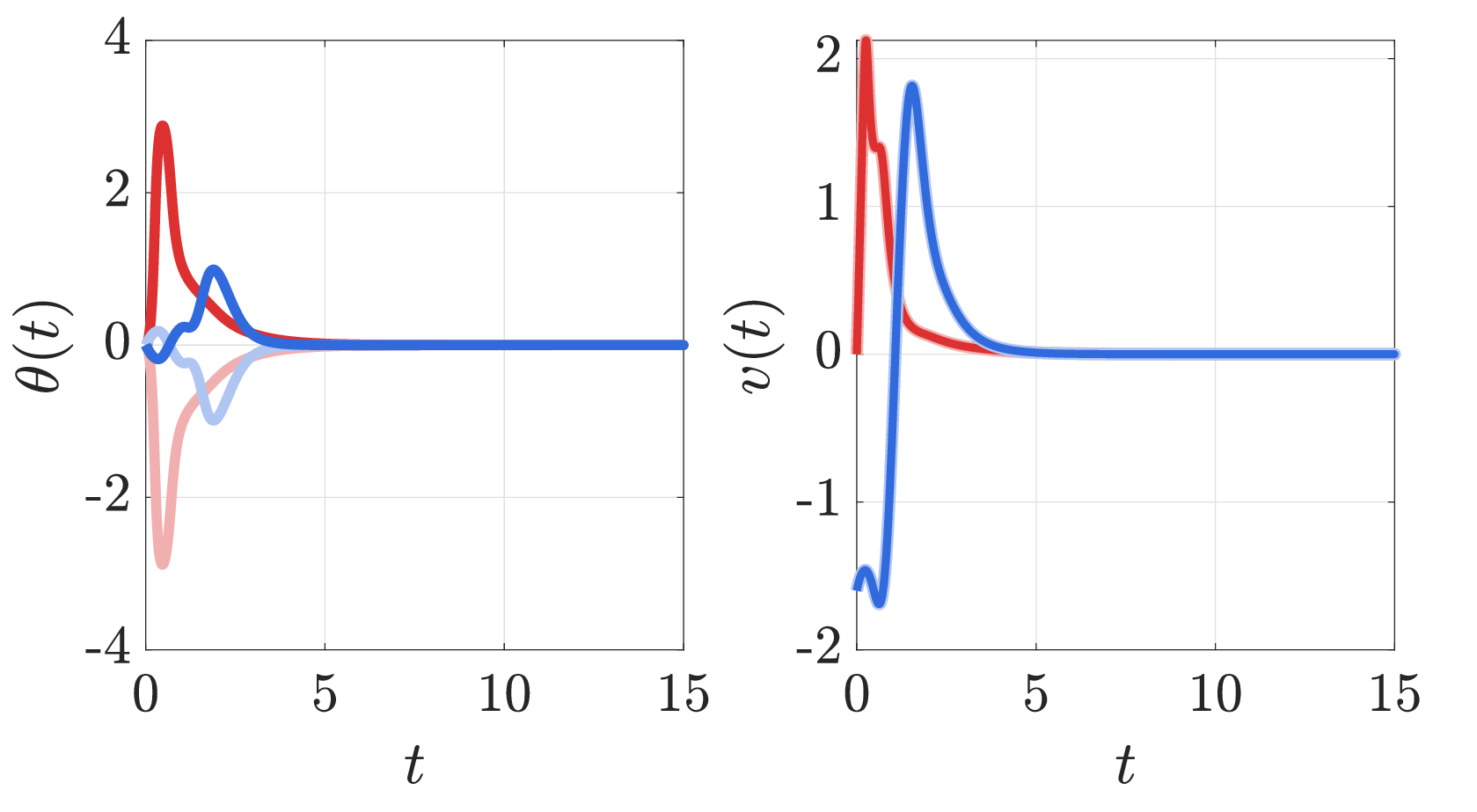}
\caption{The orientation and velocity profiles of \eqref{eq:bicycle_cartesian} with wheelbase 
$L = 4.5$, under the acceleration control law \eqref{eq:acceleration_first}, \eqref{eq:acceleration_second}, \eqref{eq:hat_a_thm} and the steering control law \eqref{eq:steering_1}, \eqref{eq:varphi_thm}.}
\label{fig:v_theta_trajectories}
\end{figure}

\section{Parallel Parking: A Benchmark for Nonholonomic Stabilization}
We illustrate the trajectories generated by the proposed control laws for a parallel parking maneuver. As shown in Fig.~\ref{fig:cartesian_pp_trajectories} and Fig.~\ref{fig:v_theta_trajectories}, the desired parking configuration is the origin, namely $x^*=y^*=\theta^*=v^*=0$ where the target heading $\theta^*$ is indicated by a black arrow. The initial conditions are $x_0=0$, $y_0=\pm 1$, and $\theta_0=0$ (initial heading indicated by an arrow of corresponding color). For each value of $y_0$, we simulate a forward maneuver with $v_0=0$ and gains $k_1=0.6$, $k_2=1.4$, $k_3=1$, $k_4=2$, together with a reverse maneuver with $v_0=-1.6$ and gains $k_1=1.2$, $k_2=1.5$, $k_3=0.4$, $k_4=0.8$. These trajectories clearly exhibit human-like parallel parking behavior. Furthermore, by \eqref{eq:backstepping_v_hat}, the speed $v$ is attracted to the positive manifold $N(k_2\delta)\rho$, which causes the feedback to favor forward motion. The inputs in Fig.~\ref{fig:inputs} remain smooth and satisfy the prescribed bounds.

\section{Conclusion}
We derive the first backstepping control laws for parking of the bicycle model of a car in polar coordinates. By introducing range-normalized coordinates that encode the parking geometry, our feedback laws achieve global exponential output stability in the polar coordinates and generate human-like parking maneuvers. This also sets the stage for safety extensions, optimality-based designs built on the developed CLF~\eqref{eq:CLF_singular}, and analogous control laws for tractor-trailer systems.

\begin{figure}[t]
    \centering
\includegraphics[width=0.85\linewidth]{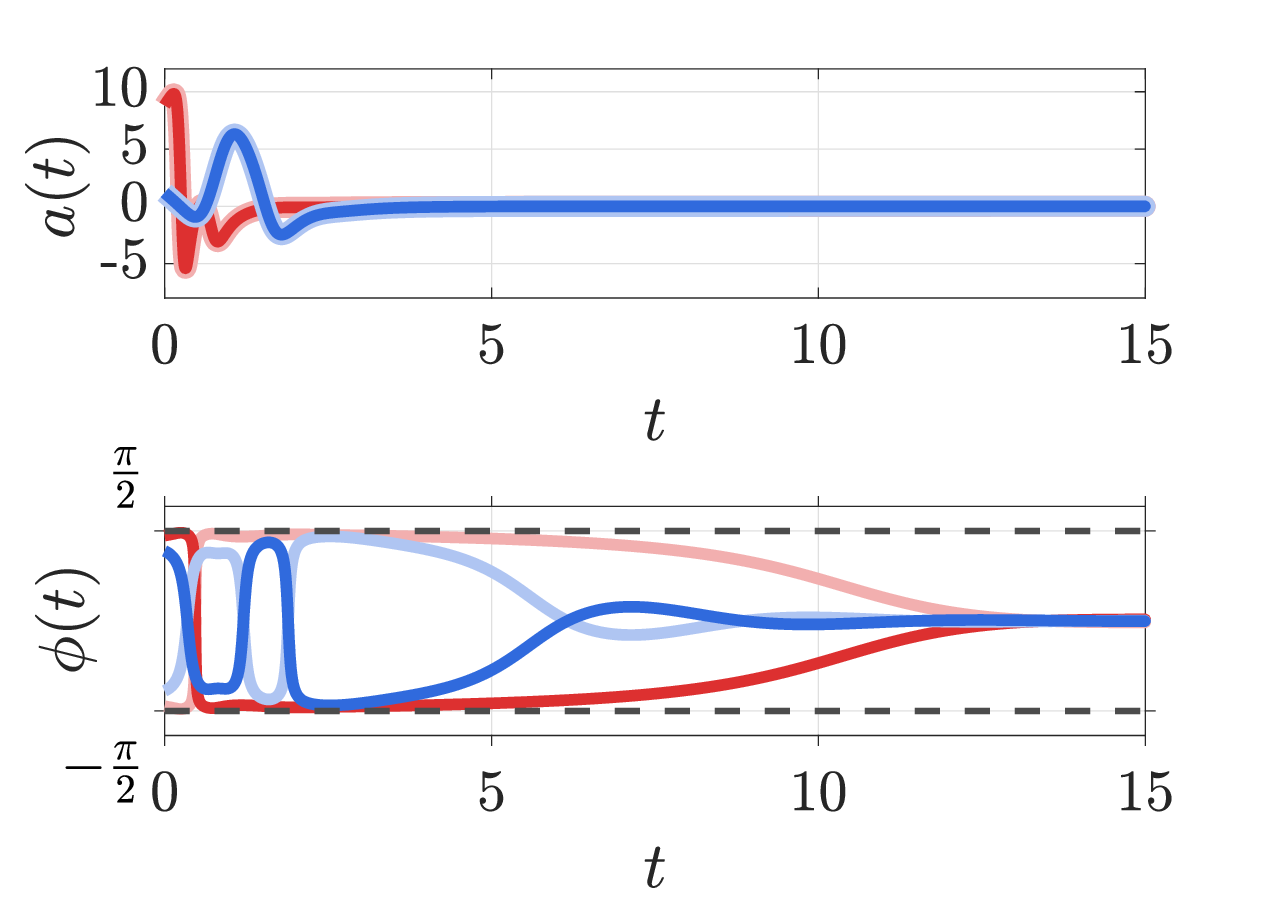}
    \caption{The longitudinal acceleration \eqref{eq:acceleration_first}, \eqref{eq:acceleration_second}, \eqref{eq:hat_a_thm} and the steering angle \eqref{eq:steering_1}, \eqref{eq:varphi_thm} with $\phi(t) = \arctan(\varphi(t))$ for \eqref{eq:bicycle_cartesian}. }
    \label{fig:inputs}
\end{figure}

\bibliographystyle{IEEEtranS}
\bibliography{root_ACC1}

\clearpage

\end{document}